\documentclass[12pt]{article}
\usepackage[usenames]{color}
\usepackage{tabularx,colortbl}
\usepackage{amsmath,amssymb,amsfonts}
\usepackage{setspace}

\newtheorem{theorem}{Theorem}
\newtheorem{corollary}{Corollary}
\newtheorem{proposition}{Proposition}
\newtheorem{lemma}{Lemma}
\newtheorem{example}{Example}
\newtheorem{definition}{Definition}

\newtheorem{prop}{Proposition}
\newtheorem{lem}{Lemma}
\newtheorem{cor}{Corollary}
\newtheorem{remark}{Remark}

\def\sds{\strut \displaystyle}

\def\R{{\mathbb R}}

\def\I{{\mathbb I}}

\newcommand{\beq}{\begin{equation}}
\newcommand{\eeq}{\end{equation}}
\newcommand{\beas}{\begin{eqnarray*}}
\newcommand{\eeas}{\end{eqnarray*}}
\newcommand{\bea}{\begin{eqnarray}}
\newcommand{\eea}{\end{eqnarray}}
\newcommand{\bei}{\begin{itemize}}
\newcommand{\eei}{\end{itemize}}
\newcommand{\ben}{\begin{enumerate}}
\newcommand{\een}{\end{enumerate}}
\newcommand{\bet}{\begin{theorem}}
\newcommand{\eet}{\end{theorem}}
\newcommand{\bel}{\begin{lemma}}
\newcommand{\eel}{\end{lemma}}
\newcommand{\bep}{\begin{proposition}}
\newcommand{\eep}{\end{proposition}}
\newcommand{\bed}{\begin{definition}}
\newcommand{\eed}{\end{definition}}
\newcommand{\bec}{\begin{corollary}}
\newcommand{\eec}{\end{corollary}}
\newcommand{\bex}{\begin{example}}
\newcommand{\eex}{\end{example}}

\newcommand{\argmin}{\mathop{\rm arg\min}}

\def\sb{{\cal B}}

\newenvironment{proof}[1][Proof]{\noindent\textbf{#1.} }{\
  \rule{0.5em}{0.5em}}

\begin{document}

\markboth{{\small\it RIP}} {{\small\it T. Cai, Y. Wang, and G. Xu}}

\title{New Bounds for Restricted Isometry Constants}

\author{T. Tony Cai\thanks{Department of Statistics, The Wharton School,
University of Pennsylvania, Philadelphia, PA 19104, USA; e-mail:
{\tt tcai@wharton.upenn.edu}. Research supported in part by NSF
Grant DMS-0604954 and NSF FRG Grant DMS-0854973.}
~ Lie Wang\thanks{Department of Mathematics,
Massachusetts Institute of Technology, Cambridge, MA 02139, USA;
e-mail: {\tt liewang@math.mit.edu}}~ and Guangwu
Xu\thanks{Department of EE \& CS, University of Wisconsin-Milwaukee,
Milwaukee, WI 53211, USA; e-mail: {\tt gxu4uwm@uwm.edu}. Research
supported in part by the National 973 Project of China (No.
2007CB807902).}}

\date{}

\maketitle

\begin{abstract}

In this paper we show that if the restricted isometry constant
$\delta_k$ of the compressed sensing matrix satisfies
\[
\delta_k < 0.307,
\]
then $k$-sparse signals are guaranteed to be recovered exactly via $\ell_1$
minimization when no noise is present and $k$-sparse signals can be
estimated stably in the noisy case. It is also shown that the bound cannot be
substantively improved. An explicitly example is constructed in which
$\delta_{k}=\frac{k-1}{2k-1} < 0.5$, but it is impossible to recover
certain $k$-sparse signals.
\end{abstract}


\section{Introduction}
Compressed sensing aims to recover high dimensional sparse signals
based on considerably fewer linear measurements. Formally one
considers the following model:
\beq
\label{model} y = \Phi\beta + z
\eeq
where the matrix $\Phi \in \R^{n\times p}$ (with $n \ll p$)  and $z\in
\R^n$ is a vector of measurement errors. The goal is to reconstruct
the unknown signal $\beta\in \R^p$ based on $y$ and $\Phi$. A
remarkable fact is that $\beta$ can be recovered exactly in the
noiseless case under suitable conditions, provided that the signal is sparse.

A na\"ive approach for solving this problem is to consider
$\ell_0$ minimization where the goal is to find the sparsest solution
in the feasible set of possible solutions. However this is NP hard and thus is
computationally infeasible. It is then natural to consider the
method of $\ell_1$ minimization which can be viewed as a convex
relaxation of $\ell_0$ minimization. The $\ell_1$ minimization method in
this context is
\beq (P_{\sb}) \quad\quad
\hat \beta = \argmin_{\gamma \in \R^p} \{\|\gamma\|_1 \; \mbox{
  subject to } \; y - \Phi\gamma\in \sb\}
\eeq
where $\sb$ is a bounded set determined by the noise structure. For example,
$\sb=\{0\}$ in the noiseless case and $\sb$ is the feasible set of the
noise in the case of bounded error. This method has
been successfully used as an effective way for reconstructing a
sparse signal in many settings. See, e.g., \cite{CanRomTao,CanTao05,CanTao06,CanTao07,Donoho,DonHuo,CaWaXu,CaWaXu1}.

One of the most commonly used frameworks for sparse recovery via $\ell_1$
minimization is the {\sl Restricted Isometry Property (RIP)}
introduced by Cand\`es and Tao \cite{CanTao05}.
RIP essentially requires that every subset of columns of $\Phi$ with
certain cardinality approximately behaves like an orthonormal system.
A vector $v=(v_i) \in
\R^p$ is {\it $k$-sparse} if $|\mbox{supp} (v)| \le k$, where $\mbox{supp} (v) = \{ i : v_i\neq
0\}$ is the support of $v$. For an $n\times p$ matrix $\Phi$ and an integer $k$, $1\le k \le p$,
the {\it $k$-restricted isometry constant} $\delta_k(\Phi)$ is the smallest constant such that
\begin{equation}\label{cond:2.1}
\sqrt{1-\delta_k(\Phi)}\|c\|_2 \le \|\Phi c\|_2 \le \sqrt{1+\delta_k(\Phi)}\|c\|_2
\end{equation}
for every $k$-sparse vector $c$. If $k+k'\le p$, the
{\it $k, k'$-restricted orthogonality constant} $\theta_{k,k'}(\Phi)$,
is the smallest number that satisfies
\begin{equation}\label{cond:2.2}
|\langle \Phi c, \Phi c'\rangle| \le \theta_{k, k'}(\Phi)\|c\|_2\|c'\|_2,
\end{equation}
for all $c$ and $c'$ such that $c$ and $c'$ are $k$-sparse and $k'$-sparse respectively, and have
disjoint supports. For notational simplicity we shall write $\delta_k$
for $\delta_k(\Phi)$ and $\theta_{k, k'}$ for $\theta_{k, k'}(\Phi)$
hereafter.

It has been shown that $\ell_1$ minimization can recover a sparse
signal with a small or zero error under various conditions on
$\delta_k$ and $\theta_{k,k'}$. For example, the condition
$\delta_k+\theta_{k,k}+\theta_{k,2k}<1$ was used in Cand\`es and Tao
\cite{CanTao05}, $\delta_{3k}+3\delta_{4k}<2$ in Cand\`es, Romberg and
Tao \cite{CanRomTao}, and $\delta_{2k}+\theta_{k,2k}<1$ in Cand\`es
and Tao \cite{CanTao07}. In \cite{CaXuZh}, Cai, Xu and Zhang proved
that stable recovery can be achieved when
$\delta_{1.5k}+\theta_{k,1.5k}<1$
\footnote{For a positive real number $\alpha$,  $\delta_{\alpha k}$
  and $\theta_{k,\alpha k}$ are understood as $\delta_{\lceil
\alpha k \rceil}$ and $\theta_{k,\lceil \alpha k \rceil}$. }.
In a recent paper, Cai, Wang and Xu \cite{CaWaXu1} further improve
the condition to  $\delta_{1.25k}+\theta_{k,1.25k}<1$.

It is important to note that RIP conditions are difficult to verify
for a given matrix $\Phi$. A widely used technique for avoiding checking
the RIP directly is to generate the matrix $\Phi$
randomly and to show that the resulting random matrix satisfies the
RIP with high probability using the well-known Johnson-Lindenstrauss
Lemma. See, for example, Baraniuk, et al. \cite{BaDaDeWa}. This is
typically done for conditions involving only the restricted isometry
constant $\delta$. Attention has been focused on $\delta_{2k}$ as it
is obviously necessary to have $\delta_{2k}<1$ for model
identifiability.  In a recent
paper, Davies and Gribonval \cite{DavGri} constructed examples which
showed that  if $\delta_{2k}\ge \frac{1}{\sqrt{2}}$,
exact recovery of certain $k$-sparse signal can fail in the noiseless case.
On the other hand, sufficient conditions on $\delta_{2k}$ has been given.
For example, $\delta_{2k} < \sqrt{2} -1$ is used by Cand\`es
\cite{Candes} and $\delta_{2k} < 0.4531$ by Fouchart and Lai \cite{FOLA}.
The results given in Cai, Wang and Xu \cite{CaWaXu1} implies that
$\delta_{2k}< 0.472$ is a sufficient condition for sparse signal recovery.

Among the conditions of the form $\delta_{ck} < \alpha$, the most
natural and desirable condition for recovering a $k$-sparse signal is
arguably
\[
\delta_{k} < \alpha,
\]
for some quantity $\alpha$.

The purpose of this paper is to establish, to the best of our
knowledge, the first such condition on $\delta_k$. To be more specific, we
show that under the condition
\beq
\label{k-bound}
\delta_{k} < 0.307,
\eeq
$k$-sparse signals are guaranteed to be recovered exactly via $\ell_1$
minimization when no noise is present and $k$-sparse signals can be
estimated stably in the noisy case. Although we are mainly interested
in recovering sparse signals, the results can be extended to the
general setting where the true signal is not necessarily $k$-sparse.

It is also shown in the present paper that the bound (\ref{k-bound}) cannot be
substantively improved. An upper bound for $\delta_{k}$ is also
given. An explicitly example is constructed in which
$\delta_{k}=\frac{k-1}{2k-1} < 0.5$, but it is impossible to recover
certain $k$-sparse signals.

Our analysis is simple and elementary.
The main ingredients in proving the new RIP conditions are the {\sl norm-inequality for $\ell_1$ and $\ell_2$}, and the {\sl square root lifting
inequality} for the restricted orthogonality constant $\theta_{k,k'}$.
Let $x\in \R^n$. A direct consequence of the Cauchy-Schwartz inequality
is that $0\le \|x\|_2-\frac{\|x\|_1}{\sqrt{n}}$. Our norm-inequality for $\ell_1$ and $\ell_2$ gives an upper bound
for the quantity $\|x\|_2-\frac{\|x\|_1}{\sqrt{n}}$, namely
\begin{equation}\label{eq:1.1}
\|x\|_2 - \frac{\|x\|_1}{\sqrt{n}}\le\frac{\sqrt{n}}{4}\big(\max_{1\le i\le n} |x_i|-\min_{1\le
i\le n} |x_i|\big).
\end{equation}
This is an inequality of its own interest. The square root lifting inequality is a result we developed in \cite{CaWaXu1} which states
that if $a \ge 1$ and $k', ak'$ are positive integers, then
\begin{equation}\label{eq:1.2}
\theta_{k, ak'} \le \sqrt{a}\theta_{k, k'}.
\end{equation}
Indeed we derive a more general result on RIP and obtain
(\ref{k-bound}) as a special case.

The paper is organized as follows. After Section \ref{sec:ripproperies}, in
which some basic properties of restricted isometry constants are
discussed, we introduce in Section \ref{sec:nmineq} a norm inequality
for $\ell_1$ and $\ell_2$, which enables us to make finer analysis of
the sparse recovery problem. Our new RIP bounds are presented in
Section \ref{sec:newbds}. In Section \ref{sec:uppbds}, upper bounds
for RIP constants are given.

\section{Some Properties of Restricted Isometry Constants}
\label{sec:ripproperies}

We begin by introducing basic notations and definitions related to the RIP.
We also collect a few
elementary results needed for the later sections.

For a vector $v=(v_i) \in \R^p$, we shall denote by $v_{\max(k)}$ the vector $v$ with all but the
$k$ largest entries (in absolute value) set to zero and define $v_{-\max(k)} = v - v_{\max(k)}$,
the vector $v$ with the $k$ largest entries (in absolute value) set to zero. We use the standard
notation $\|v\|_q=(\sum_{i=1}^p |v_i|^q)^{1/q}$ to denote the $\ell_q$-norm of the vector $v$. We
shall also treat a vector $v=(v_i)$ as a function $v: \{1,2,\cdots, p\}\rightarrow \R$ by assigning
$v(i)=v_i$.

For a subset $T$ of  $ \{1, \cdots, p\}$, we use $\Phi_T$ to denote the submatrix obtained by
taking the columns of $\Phi$ according to the indices in $T$. Let
\[
\mathcal{SSV}_T = \{\lambda
: \lambda \mbox{ an eigenvalue of }\Phi'_T\Phi_T \},
\]
 and $\sds \Lambda_{\min}(k)=\min
\{\cup_{|T|\le k}\mathcal{SSV}_T\}$, $\sds \Lambda_{\max}(k)=\max\{\cup_{|T|\le k}\mathcal{SSV}_T\}$. It
can be seen that
\[
1-\delta_{k}\le \Lambda_{\min}(k)\le \Lambda_{\max}(k)\le 1+\delta_k.
\]
Hence the condition (\ref{cond:2.1}) can be viewed as a condition on
$\Lambda_{\min}(k)$ and $\Lambda_{\max}(k)$.

The following monotone properties can be easily checked,
\begin{eqnarray}\label{eq:2.3}
&&\delta_k \le \delta_{k_1}, \mbox{ if } k\le k_1\le p\\
\label{eq:2.4} &&\theta_{k,k'}\le \theta_{k_1, k_1'}, \mbox{ if }  k\le k_1,  k'\le k_1', \mbox{
and } k_1+k_1'\le p.
\end{eqnarray}

Cand\`es and Tao \cite{CanTao05} showed that the constants $\delta_k$ and $\theta_{k,k'}$ are
related by the following inequalities,
\begin{equation}\label{eq:2.5}
\theta_{k,k'} \le \delta_{k+k'}\le \theta_{k,k'} +\max(
\delta_{k},\delta_{k'}).
\end{equation}
In the following, we list several refinements to the inequalities (\ref{eq:2.5}) whose proofs will be
provided in the appendix.
\begin{lem}\label{lem:2.0}
For any positive integers $k$ and $k'$, we have
\begin{eqnarray}
&&\delta_{k+k'}\leq \theta_{k,k'}+\frac{k\delta_k+k'\delta_{k'}}{k+k'} \label{eqn2.0-1}\\
&&\delta_{k+k'}\leq \frac{2\sqrt{k k'
}}{k+k'}\theta_{k,k'}+\max\{\delta_k,\delta_{k'}\}\label{eqn2.0-2}
\end{eqnarray}
\end{lem}

The following properties for $\delta$ and $\theta$, developed by
Cai, Xu and Zhang in \cite{CaXuZh}, have been especially useful in producing simplified recovery conditions:
\begin{equation}\label{eq:2.6}
\theta_{k, \sum_{i=1}^l k_i}\le \sqrt{\sum_{i=1}^l \theta_{k, k_i}^2}
\le \sqrt{ \sum_{i=1}^l \delta_{k+k_i}^2}.
\end{equation}
It follows from (\ref{eq:2.6}) that  for any positive integer $a$, we have $\theta_{k,ak'}\leq
\sqrt{a}\theta_{k,k'}$. This fact was further generalized by Cai, Wang and Xu in \cite{CaWaXu1} to the following
square root lifting inequality.
\begin{lem}
\emph{(Square root lifting inequality)}
\label{lem:2.1}
For any $a\ge 1$ and positive integers $k, k'$ such that $ak'$ is an integer,
\begin{equation}\label{eq:2.7}
\theta_{k,ak'}\leq \sqrt{a}\theta_{k,k'}.
\end{equation}
\end{lem}

Using the square root lifting inequality and other properties for
RIP constants we mentioned earlier, some interesting results can be
produced. For example,
\begin{cor}
\label{cor:2.1} For any integer $k\ge 1$,
\bea
\delta_{4k} &\le& 3\delta_{2k}. \label{eq:2.8}\\
\delta_{3k} &\le& \frac{1}3\delta_{k} + (\sqrt{2}+\frac{2}3 )\delta_{2k}.
\label{eq:2.9}
\eea
\end{cor}

\section{A Norm Inequality for $\ell_1$ and $\ell_2$}
\label{sec:nmineq}

In this section, we will develop a useful inequality for achieving finer
conversion between $\ell_1$-norm and  $\ell_2$-norm.

Let $x=(x_1,x_2,\cdots,x_n)\in \R^n$. A direct consequence of the
Cauchy-Schwartz inequality is that
\[
0\le \|x\|_2-\frac{\|x\|_1}{\sqrt{n}}
\]
and the equality hold if and only if $|x_1|=|x_2|=\cdots=|x_n|$. The next result reveals some information
about how large the quantity $\|x\|_2-\frac{\|x\|_1}{\sqrt{n}}$ can be.

\begin{prop}
\label{norm.ieq} For any $x \in \R^n$,
\[
\|x\|_2 - \frac{\|x\|_1}{\sqrt{n}}\le\frac{\sqrt{n}}{4}\big(\max_{1\le i\le n} |x_i|-\min_{1\le
i\le n} |x_i|\big).
\]
The equality is attained if and only if $|x_1|=|x_2|=\cdots=|x_n|$, or $n=4m$ for some positive
integer $m$ and $x$ satisfies $|x_{i_1}|=|x_{i_2}|=\cdots=|x_{i_m}|$
for some $1\le i_1 < i_2<\cdots < i_m\le n$ and
$x_{k}=0$ for $k\notin \{i_1, i_2, ..., i_m\}$.
\end{prop}

\begin{proof}
It is obvious that the result holds when the absolute values of all
coordinates are equal. Without loss of generality, we now assume that
$x_1\geq x_2\geq\cdots\geq x_n\geq 0$ and not all $x_i$ are equal. Let
$$f(x)=\|x\|_2-\frac{\|x\|_1}{\sqrt{n}}.$$
Note that for any $i\in \{2,3,\cdots,n-1\}$,
$$\frac{\partial f}{\partial
x_i}=\frac{x_i}{\|x\|_2}-\frac{1}{\sqrt{n}}.$$
This implies that when $x_i\leq \frac{\|x\|_2}{\sqrt{n}}$,
$f(x)$ is decreasing in $x_i$; otherwise $f(x)$ is increasing in $x_i$.
Therefore, if we fix $x_1$ and
$x_n$, when $f(x)$ achieves its maximum, $x$ must be of the form that $x_1=x_2=\cdots=x_k$ and
$x_{k+1}=\cdots=x_n$ for some $1\leq k < n$. Now
$$f(x)=\sqrt{k(x_1^2-x_n^2)+n
x_n^2}-\frac{k}{\sqrt{n}}(x_1-x_n)-\sqrt{n}x_n.$$
Treat this as a function of $k$ for $k\in (0,n)$
$$g(k)=\sqrt{k(x_1^2-x_n^2)+n
x_n^2}-\frac{k}{\sqrt{n}}(x_1-x_n)-\sqrt{n}x_n.$$
By taking derivative, it is easy to see that
$$g(k)\leq
g(n\frac{(\frac{x_1+x_n}{2})^2-x_n^2}{x_1^2-x_n^2})=\sqrt{n}(x_1-x_n)(\frac{1}{2}-\frac{x_1+3x_n}{4(x_1+x_n)}).$$
Now, since $\frac{x_1+3x_n}{4(x_1+x_n)}\geq 1/4$, we have
\[
\|x\|_2 \le \frac{\|x\|_1}{\sqrt{n}}+\frac{\sqrt{n}}{4}\big(x_1-x_n\big).
\]
We can also see that the above inequality becomes equality if and only
if $x_{k+1}=\cdots =x_n=0$ and $k=n/4$.
\end{proof}

\begin{remark}{\rm
A direct consequence of Proposition \ref{norm.ieq} is that for any $x \in \R^n$,
\[
\|x\|_2 \le \frac{\|x\|_1}{\sqrt{n}} + \frac{\sqrt{n}\|x\|_{\infty}}{4}.
\]
}
\end{remark}

\section{New RIP Bounds of Compressed Sensing Matrices}
\label{sec:newbds}

In this section, we consider new RIP conditions for sparse signal
recovery. However, the results can be easily extended to general
signals $\beta$ with error bounds involved with $\beta_{-\max(k)}$, as
discussed in \cite{CaWaXu1, CaXuZh}.

Suppose
\[
y=\Phi\beta + z
\]
with $\|z\|_2\le \varepsilon$. Denote $\hat\beta$
the solution of the following $\ell_1$ minimization problem,
\beq
\label{hat.beta}
\hat \beta=\argmin\{\|\gamma\|_1: \quad \mbox{subject to } \quad \|\Phi\gamma-y\|_2 \le \varepsilon\}.
\eeq

\bet
\label{thm:4.1}
Suppose $\beta$ is $k$-sparse. Let $k_1, k_2$ be positive integers
such that $k_1\ge k$ and $8(k_1-k) \le k_2$. Let
$$t=\sqrt{\frac{k_1}{k_2}}+\frac{1}{4}\sqrt{\frac{k_2}{k_1}}-\frac{2(k_1-k)}{\sqrt{k_1 k_2}}.$$
Then under the condition
\[
\delta_{k_1}+t\theta_{k_1 , k_2 } < 1,
\]
Then the $\ell_1$ minimizer $\hat \beta$ defined in (\ref{hat.beta}) satisfies
\[
\|\beta-\hat\beta\|_2\le \frac{2\sqrt{2} \sqrt{1+\delta_{k_1 }}}
{1-\delta_{ k_1 }-t\theta_{k_1, k_2}} \varepsilon.
\]
In particular, in the noiseless case where $y=\Phi\beta$, $\ell_1$
minimization recovers $\beta$ exactly.
\eet
\begin{remark}{\rm
Different choices of $k_1$ and $k_2$ can result
in different conditions. Here we list several of them which are of certain interest.\footnote{
Here we assume that the the fraction multiple of $k$  are integers. Otherwise, we have to
use the ceiling notation. }

\begin{center}
\begin{tabular}{|l|l|c|}
\hline
$k_1$ & $k_2$ & Recovery condition \\
\hline
\hline
$k$ & $k$ & $\delta_{k}+1.25\theta_{k,k}<1$\\
\hline
$k$ & ${4\over 9}k$ & $\delta_k+\frac{5}{3}\theta_{k,\frac{4k}{9}}<1$\\
\hline
${9\over 8}k$ & $k$ &$\delta_{\frac{9k}{8}}+\sqrt{9\over 8}\theta_{k, \frac{9k}{8}}<1$\\
\hline
${8\over 7}k$ & ${8\over 7}k$ & $\delta_{\frac{8k}{7}}+\theta_{\frac{8k}{7}, \frac{8k}{7}} <1$\\
\hline
\end{tabular}
\end{center}
}
\end{remark}

Now let us prove the theorem.\\
\begin{proof}
Let $h=\hat\beta - \beta$. For any subset $Q\subset \{1,2,\cdots, p\}$, we define
\[
h_Q = h\I_Q.
\]

Suppose  $|h(1)|\ge |h(2)|\ge \cdots
\ge |h(k+1)|\ge |h(k+2)| \ge \cdots$. \\Let $T=\{1,2,\cdots, k\}$ and $\Omega$ be the
support of $\beta$. The following fact, which is based on the minimality of $\hat\beta$, has been
widely used, see \cite{CaWaXu1,CanRomTao,DonHuo}.
\[
\|h_{\Omega}\|_1 \ge \| h_{\Omega^c}\|_1.
\]
It is obvious that $\|h_{\Omega^c\cap T}\|_1 \ge \|h_{\Omega^c\cap T^c}\|_1$, so
we have
\[
\|h_T\|_1 \ge \|h_{T^{c}}\|_1.
\]

Partition $\{1,2,\cdots, p\}$ into the following sets:
\[
S_0=\{1,2,\cdots, k_1\}, S_1=\{k_1+1, \cdots, k_1+k_2\},
S_2=\{k_1+k_2+1, \cdots, k_1+2k_2\}, \cdots.
\]

Then it follows from Proposition~\ref{norm.ieq} that
\begin{eqnarray*}
\sum_{i\geq 1}\|h_{S_i}\|_2&\leq&
\sum_{i\ge 1}\frac{\|h_{S_i}\|_1}{\sqrt{k_2}}+\\
&&\frac{\sqrt{k_2}}{4}\big(|h(k_1+1)|- |h(k_1+k_2)|+|h(k_1+k_2+1)|-|h(k_1+2k_2)|+\cdots \big)\\
&\leq&\frac{1}{\sqrt{k_2}}(\|h_{T^{c}}\|_1-(k_1-k)|h(k_1+1)|)+\frac{\sqrt{k_2}}{4}|h(k_1+1)|\\
&\leq&\frac{1}{\sqrt{k_2}}(\|h_{T}\|_1-(k_1-k)|h(k_1+1)|)+\frac{\sqrt{k_2}}{4}|h(k_1+1)|\\
&\leq&\frac{1}{\sqrt{k_2}}(\|h_{S_0}\|_1-2(k_1-k)|h(k_1+1)|)+\frac{\sqrt{k_2}}{4}|h(k_1+1)|\\
&\leq&\sqrt{\frac{k_1}{k_2}}\|h_{S_0}\|_2+\left(\frac{\sqrt{k_2}}{4} - \frac{2(k_1-k)}{\sqrt{k_2}}\right)|h(k_1+1)|\\
&\leq&\sqrt{\frac{k_1}{k_2}}\|h_{S_0}\|_2+\left(\frac{\sqrt{k_2}}{4\sqrt{k_1}} - \frac{2(k_1-k)}{\sqrt{k_1
k_2}}\right)\|h_{S_0}\|_2=t\|h_{S_0}\|_2
\end{eqnarray*}
Now
\begin{eqnarray*}
|\langle \Phi h, \Phi h_{S_0}\rangle| &=& |\langle\Phi h_{S_0}, \Phi
h_{S_0}\rangle+\sum_{i\ge
1}\langle\Phi_{S_i}h_{S_i}, \Phi h_{S_0}\rangle|\\
  &\ge& (1-\delta_{k_1})\|h_{S_0}\|_2^2 -\theta_{k_1, k_2}\|h_{S_0}\|_2\sum_{i\ge
  1}\|h_{S_i}\|_2\\
&\ge& (1-\delta_{k_1}-t\theta_{k_1, k_2})\|h_{S_0}\|_2^2
\end{eqnarray*}
Note that
\[
\|\Phi h\|_2=\|\Phi(\beta-\hat{\beta})\|_2\le \|\Phi\beta-y\|_2+
\|\Phi\hat{\beta}-y\|_2\le 2\varepsilon.
\]
Also the next relation
\begin{eqnarray*}
\|h_{S_0^c}\|_2^2\leq \|h_{S_0^c}\|_1 \frac{\|h_{S_0}\|_1}{k_1}\leq
\frac{\|h_{S_0}\|_1^2}{k_1}\leq \|h_{S_0}\|_2^2
\end{eqnarray*}
implies
\[
\|h\|_2^2 = \|h_{S_0}\|_2^2+ \|h_{S_0^c}\|_2^2 \le 2\|h_{S_0}\|_2^2.
\]

Putting them together we get\footnote{If $h_{S_0}= 0$, then the theorem is trivially true.
So here  we assume that $h_{S_0}\neq 0$.}
\begin{eqnarray*}
\|h\|_2 &\le& \sqrt{2}\|h_{S_0}\|_2\\
&\le&  \frac{\sqrt{2}|\langle \Phi h, \Phi
h_{S_0}\rangle|}{(1-\delta_{k_1}-t \theta_{k_1,
k_2})\|h_{S_0}\|_2}\\
&\le&\frac{\sqrt{2}\|\Phi h\|_2\|\Phi h_{S_0}\|_2}
{(1-\delta_{k_1}- t \theta_{k_1, k_2})\|h_{S_0}\|_2}\\
&\le &
\frac{2\sqrt{2}\varepsilon \sqrt{1+\delta_{k_1}}\|h_{S_0}\|_2}
{(1-\delta_{k_1}-t\theta_{k_1, k_2})\|h_{S_0}\|_2}\\
&\le& \frac{2\sqrt{2}\sqrt{1+\delta_{k_1}}}
{1-\delta_{k_1}-t \theta_{k_1, k_2}}\varepsilon.
\end{eqnarray*}
\end{proof}

The following is our main result of the paper. It is the consequence
of Theorem~\ref{thm:4.1} and the square root lifting inequality.
\bet
\label{thm:4.2}
Let $y=\Phi\beta + z$ with $\|z\|_2\le \varepsilon$. Suppose
$\beta$ is $k$-sparse with $k > 1$. Then under the condition
\[
\delta_{k} < 0.307
\]
the constrained $\ell_1$ minimizer $\hat \beta$ given in
(\ref{hat.beta}) satisfies
\[
\|\beta-\hat\beta\|_2\le \frac{\varepsilon}{0.307 - \delta_{k}},
\]
In particular, in the noiseless case $\hat \beta$ recovers $\beta$ exactly.
\eet

To the best of our knowledge, this seems to be the first result for sparse recovery
with conditions that only involve $\delta_k$.

\

\begin{proof}
We will present the proof for the case $k\equiv 0 \pmod 9$ in this
section. This is the case that can be treated in a concise way and for
which the proof also conveys the main ideas. The complete proof will be
given in the appendix.

In Theorem \ref{thm:4.1}, set $k_1=k$ and $k_2=\frac{4}{9}k $. Let
\[
t = \sqrt{\frac{k}{k_2}} +\frac{1}4\sqrt{\frac{k_2}{k}}=\frac{5}3.
\]
Then under the condition
$$\delta_k+\frac{5}3\theta_{k,\frac{4}{9}k}<1$$
we have
\[
\|\beta-\hat\beta\|_2\le \frac{2\sqrt{2} \sqrt{1+\delta_{k }}}
{1-\delta_{ k }-\frac{5}3\theta_{k,  \frac{4}{9}k}} \epsilon.
\]

Using the square root lifting inequality, we get
\begin{eqnarray*}
\delta_{k} + \frac{5}3\theta_{k,\frac{4k}9} & = & \delta_{k} + \frac{5}3\theta_{\frac{9}5\frac{5k}9,\frac{4k}9}\\
&\le &\delta_{k} + \frac{5}3\sqrt{\frac{9}5} \theta_{\frac{5k}9,\frac{4k}9}\le (1+\sqrt{5})\delta_k\\
&<& 1.
\end{eqnarray*}
In this case,
\begin{eqnarray*}
\|\beta-\hat\beta\|_2 &\le & \frac{2\sqrt{2} \sqrt{1+\delta_{k }}}
{1-\delta_{ k }-t\theta_{k, \frac{4k}9}} \varepsilon \le \frac{2\sqrt{2} \sqrt{1+\delta_{k }}}{1-(1+\sqrt{5})\delta_k}\varepsilon\\
&\le & \frac{3.256}{1-3.256 \delta_k}\varepsilon
\le \frac{\varepsilon}{0.307 - \delta_{k}}.
\end{eqnarray*}
\end{proof}

\begin{remark}
{\rm
\begin{enumerate}
\item It can be seen from the proof that we actually have a slightly better estimation, that is,
\[
\|\beta-\hat\beta\|_2\le \frac{2\sqrt{2} \sqrt{1+\delta_{k
}}}{1-C_0\cdot \delta_{k}} \varepsilon,
\]
where $C_0 = 1+\frac{23}{2\sqrt{26}}< 3.256.$
\item For simplicity, we have focused on recovering $k$-sparse signals in
the present paper. When $\beta$ is not $k$-sparse, $\ell_1$
minimization can also recover $\beta$ with accuracy if $\beta$ has
good $k$-term approximation. Similar to \cite{CaWaXu,CaXuZh}, this result can be extended to the general setting.
Under the condition $\delta_k < 0.307$,
Theorem \ref{thm:4.2} holds with the error bound
\[
\|\hat\beta - \beta\|_2 \le \frac{\varepsilon}{0.307 - \delta_{k}}+ \frac{1}{0.307 - \delta_{k}}\frac{\|\beta_{-\max(k)}\|_1}{\sqrt{k}}.
\]
\end{enumerate}
}
\end{remark}

We now consider stable recovery of $k$-sparse signals with error in a
different bounded set.
Cand\`es and Tao \cite{CanTao07} treated the sparse signal recovery
in the Gaussian noise case by  solving $(P_{\sb})$ with
$\sb = \sb^{DS} =\{z : \|\Phi' z\|_{\infty}\le \eta\}$ and referred the
solution as the {\sl Dantzig Selector}. The following result shows
that the condition $\delta_k < 0.307$ is also sufficient when
the error is in the bounded set $\sb^{DS} =\{z : \|\Phi'
z\|_{\infty}\le \lambda\}$.
\bet
\label{DS.thm}
Consider the model (\ref{model}) with $z$ satisfying
$\|\Phi' z\|_{\infty}\le \lambda$. Suppose $\beta$ is $k$-sparse and
 $\hat\beta$ is the minimizer
$$\hat \beta =\argmin_{\gamma\in\R^p}\{\|\gamma\|_1 : \|\Phi'
  (y-\Phi\gamma)\|_{\infty}\le \lambda\}.$$
Then
\[
\|\hat\beta - \beta\|_2 \le \frac{\sqrt{k}}{0.307-\delta_k}\lambda
\]
\eet
The proof of this theorem can be easily obtained based on a
minor modification of the proof of Theorem~\ref{thm:4.1}.

\section{Upper Bounds of $\delta_k$}
\label{sec:uppbds}

We have established the sparse recovery condition
\[
\delta_k < 0.307
\]
in the previous section. It is interesting to know the limit
of possible improvement within this framework. In this section,
we shall show that this bound cannot be substantively improved.
An explicitly example is constructed in which
$\delta_{k}=\frac{k-1}{2k-1} < 0.5$, but it is impossible to recover
certain $k$-sparse signals. Therefore,
the bound for $\delta_k$ cannot go beyond $0.5$ in order to guarantee
stable recovery of $k$-sparse signals.

This question was considered for the case of $\delta_{2k}$. In \cite{CaWaXu1}, among a family of recovery
conditions, it is shown  that
\[
\delta_{2k} < 0.472
\]
is sufficient for reconstructing $k$-sparse signals. On the other hand, the results of Davies and  Gribonval
\cite{DavGri} indicate that $\frac{1}{\sqrt{2}}\approx 0.707$ is likely the upper bound for $\delta_{2k}$.

\bet
\label{thm:5.1}
Let $k$ be a positive integer. Then there exists a $(2k-1)\times 2k$ matrix
$\Phi$ with the restricted isometry constant $\delta_k = \frac{k-1}{2k-1}$, and two nonzero
$k$-sparse vectors
$\beta_1$ and $\beta_2$ with disjoint supports such that
\[
\Phi \beta_1 = \Phi \beta_2.
\]
\eet
\begin{remark}{\rm This result implies that the model (\ref{model}) is
not identifiable in general under the condition
$\delta_k = \frac{k-1}{2k-1}$
and therefore not all $k$-sparse signals can be recovered exactly in
the noiseless case. In the noisy case, it is easy to see that
Theorem \ref{thm:4.2} fails because no estimator
$\hat\beta$ can be close to both $\beta_1$ and $\beta_2$ when the noisy level
 $\varepsilon$ is sufficiently small.
}
\end{remark}

\begin{proof}
Let $\Gamma$ be a $2k\times 2k$ matrix such that each diagonal element
of $\Gamma$ is 1 and each off diagonal element equals
$-\frac{1}{2k-1}$. Then it is easy to see that $\Gamma$ is a
positive-semidefinite matrix with rank $2k-1$.

Note that the symmetric matrix $\Gamma$ can be decomposed as $\Gamma=\Phi'\Phi$ where $\Phi$ is a
$(2k-1)\times 2k$ matrix with rank $2k-1$. More precisely, since $\Gamma$ has two distinct
eigenvalues $\sds \frac{2k}{2k-1}$ and $0$, with the multiplicities of $2k-1$ and $1$ respectively,
there is an orthogonal matrix $U$ such that
\[
\Gamma = U \mbox{Diag}\big\{\underbrace{\frac{2k}{2k-1},\frac{2k}{2k-1},\cdots,
\frac{2k}{2k-1}}_{2k-1}, 0\big\} U'.
\]
Define $\Phi$ as
\[
\Phi= \begin{pmatrix} \sqrt{\frac{2k}{2k-1}} & 0 & \cdots & 0 & 0 \\
0 & \sqrt{\frac{2k}{2k-1}} & \cdots & 0 & 0 \\
  &                        & \ddots & &\\
  0 & 0 &\cdots & \sqrt{\frac{2k}{2k-1}}&0\\
\end{pmatrix} U'.
\]

Let $T\subset \{1,2,\cdots, 2k\}$ with $|T|=k$. Then it can be verified that
\[
\Phi_T'\Phi_T = \begin{pmatrix} 1 & -\frac{1}{2k-1} & \cdots & -\frac{1}{2k-1} \\
 -\frac{1}{2k-1}& 1  & \cdots & -\frac{1}{2k-1} \\
  &                        & \ddots & &\\
  -\frac{1}{2k-1} & -\frac{1}{2k-1} &\cdots & 1 \\
\end{pmatrix}_{k\times k}.
\]
The characteristic polynomial of $\Phi_T'\Phi_T$ is
\[
p(\lambda) = \left( \lambda-\frac{k}{2k-1}\right)\left(\lambda-\frac{2k}{2k-1}\right)^{k-1}.
\]

This shows that for $\Phi$,
\[
\delta_k(\Phi) =1- \frac{k}{2k-1}=\frac{k-1}{2k-1}.
\]

Since the rank of $\Phi$ is $2k-1$, there exists some $\gamma\in \R^{2k}$
such that $\gamma\neq 0$ and $\Phi\gamma=0$.
Suppose $\beta_1, \; \beta_2\in
\R^{2k}$ are given by
$$\beta_1=(\gamma(1),\gamma(2),\cdots,\gamma(k),0,,\cdots,0)',$$
and
$$\beta_2=(\underbrace{0,0,\cdots,0}_{k},-\gamma(k+1),-\gamma(k+2),\cdots,-\gamma(2k))'.$$

Then both $\beta_1$ and $\beta_2$ are $k$-sparse vectors but $\Phi\beta_1=\Phi\beta_2$. This means
the model is not identifiable within the class of $k$-sparse signals.
\end{proof}

\appendix
\renewcommand{\thesubsection}{A-\arabic{subsection}}
\setcounter{equation}{0}  
\section*{APPENDIX}  
\subsection{Proof of Lemma \ref{lem:2.0}}
\begin{proof}
Let us start with the proof of (\ref{eqn2.0-1}).
We just need to show that for any $k+k'$ sparse vector $c\in \R^p$,
$$(1-\theta_{k,k'}-\frac{k\delta_k+k'\delta_{k'}}{k+k'})\|c\|_2^2\leq \|\Phi c\|_2^2\leq
(1+\theta_{k,k'}+\frac{k\delta_k+k'\delta_{k'}}{k+k'})\|c\|_2^2.$$
Assume, without loss of generality, that
$c=(c(1),c(2),\cdots,c(k+k'),0,0\cdots,0)$ and
$$|c(1)|\geq |c(2)|\geq\cdots\geq |c(k+k')|.$$
We may also assume $k\leq k'$.

Let
$c_1=(c(1),\cdots,c(k),0,0,\cdots,0)$ and $c_2=c-c_1$. It is easy
to see that $$\|c_1\|_2^2\geq \frac{k}{k+k'}\|c\|_2^2.$$ Now
\begin{eqnarray*}
\|\Phi c\|_2^2&=&\|\Phi c_1\|_2^2+\|\Phi c_2\|_2^2+2<\Phi c_1,\Phi c_2>\\
&\ge &(1-\delta_k)\|c_1\|_2^2+(1-\delta_{k'})\|c_2\|_2^2-2\theta_{k,k'}\|c_1\|_2
\|c_2\|_2\\
&\ge& \left((1-\delta_k)\frac{\|c_1\|_2^2}{\|c\|_2^2}+(1-\delta_{k'})
(1-\frac{\|c_1\|_2^2}{\|c\|_2^2})\right)\|c\|_2^2-\theta_{k,k'}(\|c_1\|_2^2+\|c_2\|_2^2)\\
&\ge &(1- \frac{k\delta_k+k'\delta_{k'}}{k+k'})\|c\|_2^2- \theta_{k,k'}\|c\|_2^2.
\end{eqnarray*}
The last inequality is due to the fact that
$\delta_k\leq\delta_{k'}$ and $\|c_1\|_2^2\geq
\frac{k}{k+k'}\|c\|_2^2$.

We can prove the upper bound by in a similar manner.

Next, we prove (\ref{eqn2.0-2}).
We just need to show that for any $k+k'$ sparse vector $c\in \R^p$,
$$(1-\frac{2\sqrt{k k'
}}{k+k'}\theta_{k,k'}-\max\{\delta_k,\delta_{k'}\})\|c\|_2^2\leq
\|\Phi c\|_2^2\leq (1+\frac{2\sqrt{k k'
}}{k+k'}\theta_{k,k'}+\max\{\delta_k,\delta_{k'}\})\|c\|_2^2.$$
We make the same arrangement of $c$ as in the proof of (\ref{eqn2.0-1}).
Then let

$c_2=(c(1),\cdots,c(k'),0,0,\cdots,0)$ and $c_1=c-c_2$. It is easy
to see that $$\|c_1\|_2^2\leq \frac{k}{k+k'}\|c\|_2^2.$$
Now
\begin{eqnarray*}
\|\Phi c\|_2^2&=&\|\Phi c_1\|_2^2+\|\Phi c_2\|_2^2+2<\Phi c_1,\Phi c_2>\\
&\leq&(1+\delta_k)\|c_1\|_2^2+(1+\delta_{k'})\|c_2\|_2^2+2\theta_{k,k'}\|c_1\|_2
\|c_2\|_2\\
&\leq&(1+\max\{\delta_k,\delta_{k'}\})\|c\|_2^2+2\theta_{k,k'}\|c\|_2^2(\frac{\|c_1\|_2}{\|c\|_2}\frac{\|c_2\|_2}{\|c\|_2})\\
&\leq&(1+\max\{\delta_k,\delta_{k'}\})\|c\|_2^2+\theta_{k,k'}\frac{2\sqrt{k
k' }}{k+k'}\|c\|_2^2.
\end{eqnarray*}
The last inequality is because $\|c_1\|_2^2\leq
\frac{k}{k+k'}\|c\|_2^2\leq \frac{1}{2}\|c\|_2^2$. The lower bound
can be proved by similar argument.
\end{proof}

\subsection{Proof of Corollary \ref{cor:2.1}}
\begin{proof}
From (\ref{eq:2.5}) and the square root lifting inequality, we have
\begin{eqnarray*}
\delta_{4k}&\le & \theta_{2k,2k}+\delta_{2k}\\
&\le & \sqrt{2}\sqrt{2}\theta_{k,k}+\delta_{2k}\\
&\le &3\delta_{2k}.
\end{eqnarray*}

By Lemma \ref{lem:2.0}, we have
\begin{eqnarray*}
\delta_{3k}&\le & \theta_{2k, k}+\frac{2\delta_{2k}+\delta_{k}}{3}\\
&\le & \frac{1}3\delta_{k} + (\sqrt{2}+\frac{2}3 )\delta_{2k}.
\end{eqnarray*}

\end{proof}
\subsection{Completion of the Proof of Theorem \ref{thm:4.2}}
\begin{proof}
In Theorem \ref{thm:4.1}, let $k_1=k, 1 \le k_2 < k$, and
\[
t = \sqrt{\frac{k}{k_2}} +\frac{1}4\sqrt{\frac{k_2}{k}} .
\]
Then under the condition
$$\delta_k+t\theta_{k,k_2}<1$$
we have
\[
\|\beta-\hat\beta\|_2\le \frac{2\sqrt{2} \sqrt{1+\delta_{k }}}
{1-\delta_{ k }-t\theta_{k, k_2}} \epsilon.
\]

By the square root lifting inequality,
\begin{eqnarray*}
\delta_k+t\theta_{k,k_2}&\leq& \delta_k+
t\theta_{\frac{k}{k-k_2}(k-k_2), k_2}\notag
\\
&\le& \delta_k+ t\sqrt{\frac{k}{k-k_2}} \theta_{k-k_2, k_2}\notag \\
&\leq&\left(1+t\sqrt{\frac{k}{k-k_2}}\right)\delta_k.
\end{eqnarray*}

Denote $A_k = 1+t\sqrt{\frac{k}{k-k_2}}$ and let
\[
f(x) = 1+\frac{1}{\sqrt{1-x}}\left(\frac{1}{\sqrt{x}}+\frac{1}4 \sqrt{x}\right)\quad x\in (0,1),
\]
then
\[
A_k = f(\frac{k_2}k).
\]
Since $\sds f'(x)=\frac{9x-4}{8(x-x^2)^{\frac{3}2}}$, $f$ is increasing
when $\frac{4}{9}\leq x < 1$ and decreasing $0< x <\frac{4}{9}$.

Let $0 \le r_k \le 8$ be the integer such that $r_k\equiv 4k \pmod 9$. Now we choose $k_2$ specifically as
follows:
\[
k_2 = \left\{
\begin{array}{ll} \lfloor \frac{4}{9}k \rfloor  &\quad \mbox{ if } r_k \le 4,\\
\lceil \frac{4}{9}k\rceil &\quad \mbox{ if } r_k > 4.
\end{array}\right.
\]
By the definition of $k_2$ we get immediately that
\[
A_k \le \max\left(f\big(\frac{4}9+\frac{4}{9k}\big), f\big(\frac{4}9-\frac{4}{9k}\big)\right).
\]
In particular, when $ k\ge 7$,
\[
A_k \le f\big(\frac{8}{21}\big)=1+\frac{23}{2\sqrt{26}}< 3.256.
\]

A direct calculation shows that
\[
A_4 = A_6 = f(0.5) = 3.25, \mbox{ and } A_5 = f(0.4) < 3.246.
\]

In order to estimate $A_k$ for $k=2, 3$, we note that in these cases $k_2=1$ and $t=\sqrt{k}$.
This is based on the observation that in
the proof of Theorem \ref{thm:4.1}, $h(k_1+(i-1)k_2+1)=h(k_1+ik_2)$ for $i>0$.
So
\[
A_2 = 1+\sqrt{2}\sqrt{2}=3, A_3 = 1+\sqrt{3}\sqrt{\frac{3}2}=3.122.
\]

These yield
\begin{eqnarray*}
\delta_k+t\theta_{k,\lceil \frac{4}{9}k\rceil}&\leq& A_k\delta_k
\\
&\le&  3.256\cdot \delta_k < 1.
\end{eqnarray*}
With the above relation, we can also get
\begin{eqnarray*}
\|\beta-\hat\beta\|_2 &\le & \frac{2\sqrt{2} \sqrt{1+\delta_{k }}}
{1-\delta_{ k }-t\theta_{k, k_2}} \varepsilon \le \frac{2\sqrt{2} \sqrt{1+\delta_{k
}}}{1-3.256\cdot \delta_{k}}\varepsilon\\
&\le & \frac{3.256}{1-3.256 \delta_k}\varepsilon
\le \frac{\varepsilon}{0.307 - \delta_{k}}.
\end{eqnarray*}
The theorem is proved.

\end{proof}

\end{document}